\documentclass{article}
\usepackage{amsmath}
\usepackage{amssymb}
\usepackage{amsthm}  
\usepackage{graphicx}

\usepackage{setspace}

\newtheorem{theorem}{Theorem}
\newtheorem{lemma}{Lemma}
\newtheorem{definition}{Definition}
\newtheorem{corollary}{Corollary}

\newcommand{\ket}[1]{\left\vert{#1}\right\rangle}
\newcommand{\bracket}[3]{\left\langle #1 \big| #2 \big| #3 \right\rangle}       

\begin{document}

\title{Quantum Query as a State Decomposition}

\author{S. A. Grillo$^{1}$ and F. L. Marquezino$^{1,2}$\\[1em]
  $^{1}$PESC/Coppe and $^{2}$Numpex-Comp/Xerem\\
Universidade Federal do Rio de Janeiro\\
Rio de Janeiro, Brazil\\[1em]
\{sgrillo, franklin\}@cos.ufrj.br}

\maketitle

\begin{abstract}
\noindent
The Quantum Query Model is a framework that allows us to express most
  known quantum algorithms. Algorithms represented by this model
  consist on a set of unitary operators acting over a finite Hilbert
  space, and a final measurement step consisting on a set of
  projectors. In this work, we prove that the application of these
  unitary operators before the measurement step is equivalent to
  decomposing a unit vector into a sum of vectors and then inverting
  some of their relative phases. We also prove that the vectors of
  that sum must fulfill a list of properties and we call such vectors
  a \emph{Block Set}. If we define the measurement step
   for the Block Set Formulation similarly to the Quantum Query Model, then we prove
  that both formulations give the same Gram matrix of output states,
  although the Block Set Formulation allows a much more explicit
  form. Therefore, the Block Set reformulation of the Quantum Query
  Model gives us an alternative interpretation on how quantum
  algorithms works.  Finally, we apply our approach to the analysis and  complexity of quantum
  exact algorithms.
\smallskip

\noindent
{\bf Keywords:} quantum exact algorithms, quantum query
  complexity, computational complexity, analysis of algorithms, design
  of algorithms.
\end{abstract}

\onehalfspacing   

\section{Introduction} 

The Quantum Query Model (QQM) is an important tool in the analysis and
design of quantum algorithms, especially because its simplicity 
allows us to compare classical 
and quantum computing more easily. This model generalizes decision trees
\cite{Dewolf} with complexity being defined as the minimum number of
oracle queries required for computing a given function $f$ for any
input~${x\in \{ 0,1\}^{n}}$.
  
The topic of exact quantum algorithms is less understood than bounded-error algorithms. For many years, the only exact
  quantum algorithms known to produce a speed-up over classical
  algorithms for total functions were those that used Deutsch's
  algorithm as a subroutine~\cite{AMB1}. The numerical method proposed by Barnum \textit{et al.}~\cite{BARNUM} just gives us approximate solutions, whose results can be laborious to translate into analytically defined algorithms for the exact case~\cite{MONTANARO}.  Currently, there is a limited number of research 
  papers that presents results in the analytic construction of exact quantum algorithms~\cite{AMB1,MONTANARO,AMB2,Ambainis4,Gruska,Ambainis5}.
  
In query complexity, the polynomial method~\cite{Beals} and adversary methods \cite{AMB3,Hoyer2} are well known for computing lower-bounds of exact quantum algorithms. There are important results about exact quantum query complexity in the literature obtained from such methods: 
\begin{itemize}
\item For an exact quantum algorithm that computes a total Boolean function within $t$ queries, there is a classical deterministic algorithm that computes the same function by applying $\mathcal{O}\left(t^{3}\right)$ queries \cite{Midrijanis}.
\item Exact quantum algorithms give an advantage for all Boolean functions excepting $AND_{n}$~\cite{Ambainis6}.
\end{itemize}

In our present work, we propose another
reformulation of the QQM  and we apply it to the analysis and complexity of exact quantum algorithms.
The new model proposed in this paper, called \textit{Block Set
  Formulation} (BSF), is shown to be equivalent to the QQM. In this
formulation, the algorithm is represented by a set of vectors
satisfying certain properties, and the unitary operators are replaced
by phase inversions on some of those vectors. This set of
vectors, called \textit{Block Set}, gives an alternative
interpretation on how quantum algorithms work. 
For each input, the BSF constructs a corresponding output
state following a different definition to the QQM. After applying the
measurement step on such output state, however, the results are identical in both models.  
The equivalence between BSF and QQM 
is proved on two steps. First, we prove that for each QQM algorithm of
$t$ queries there is a unique $t$-dimensional Block Set with the same
Gram matrix of output states. Then, we prove that for each
$t$-dimensional Block Set there are several QQM algorithms of $t$
queries with the same Gram matrix of output states. Considering that
two algorithms with the same Gram matrix of output states are
similar---since we can choose the measurement operators
appropriately---then the QQM and the BSF are equivalent.

The BSF can be simplified by proving that a BSF restricted to real
numbers is equivalent to a complex BSF. Assuming a real-valued BSF, we
prove that the Gram matrix of output states is equal to a sum of
matrices, where each of them depends on some pair of elements in the
Block Set. By using the relation between Block Sets and the final Gram
matrix we obtain a necessary and sufficient condition for the
existence of exact quantum algorithms, this condition being formulated
by a system of equations that requires a semi-definite solution. If we
consider a special case of Block Sets in which all elements are
pairwise orthogonal, we obtain a second condition for exact quantum
algorithms from our first system. This second condition is just
sufficient for exact quantum algorithms; nonetheless, we show that it
can also be used as an analytic tool for constructing exact quantum
algorithms. As an example of the application of orthogonal Block Sets,
we define the \textit{XOR-Weighted-Problem} and prove that it can be solved by exact quantum algorithms in the BSF, in which case we also give an upper-bound for its complexity.  Finally, we present a lower-bound for the exact quantum query complexity  of functions with Boolean domain and arbitrary output.

The structure of our paper is as follows. In Sec.~\ref{S1}, 
we briefly review the basic concepts and formulations of the Quantum Query Model.  In
Sec.~\ref{S2}, %
we introduce the Block Set Formulation and prove its equivalence to the Quantum Query Model.  In Sec.~\ref{secpbs},
we present the relation between Block Sets and the Gram matrix of
output states.  In Sec.~\ref{S5}, %
we show a linear condition for quantum exact algorithms and obtain a model that characterizes the computing power of a family of QQM algorithms. In Sec.~\ref{S6}, we prove a lower-bound for exact quantum query algorithms. In Sec.~\ref{S7}, we present our conclusion and discuss
potential extensions of this approach. At last, in the appendices, we
present examples.

\section{Preliminaries}
\label{S1}

Let $H$ be a finite Hilbert space and let $T$ be a finite set. Two
operators $A$ and $B$ are orthogonal if $\bracket{ \Psi }{ A^\dagger
  B}{ \Phi } =0$ for all $\ket{\Phi}, \ket{ \Psi } \in H$.  A Complete
Set of Orthogonal Projectors (CSOP) is an indexed set of pairwise
orthogonal projectors $\left\{ P_{z}:z\in T\right\}$, satisfying
\begin{equation}
  \sum_{z\in T}P_{z}=I_{H},
\end{equation}
where $I_{H}$ is the identity operator on $H$. We denote $H_{z}^{P}$ and $H_{z}^{Q}$, as the ranges of spaces projected 
by $P_{z}$ and $Q_{z}$ respectively.

\begin{lemma}
  \label{lem1}
  If $\left\{ P_{z}:z\in T\right\}$ is a CSOP and $U$ is a unitary
  operator, then $$\left\{ U^{\dagger}P_{z}U:z\in T\right\}$$ is a CSOP.
\end{lemma}

\begin{lemma}
  \label{lem2}
  Let $\left\{ P_{z}:z\in T\right\}$ and $\left\{Q_{z}:z\in T\right\}$
  be two CSOP  over  $H$ such that
$\dim\left(H_{z}^{P}\right)=\dim\left(H_{z}^{Q}\right)$  for  the same
  $z$.  Then, there exists a unitary operator $U$ such that
  $U^{\dagger}P_{z}U=Q_{z}$ for all $z\in T$.
\end{lemma}

The Quantum Query Model (QQM) is a formulation that
  simplifies the analysis of quantum algorithms for computing a function $f:\{0,1\}^n \rightarrow \{0,1\}$ for an input
  $x$ making \textit{queries} to the values $x_i$. In this model, we
  are mostly concerned with the number of queries to the
  input.
In the QQM, the memory can be divided in two registers: (i) the
\textit{query register}, whose size should allow it to represent any
integer $i\in\left\{ 0,..,n\right\}$, for an input of size $n$; and
(ii) the \textit{working memory}, without size constraints.

The query register and the working memory are jointly known as the
\textit{accessible memory}. The computational basis for the associated
Hilbert space $H_{A}$ (or, \textit{accessible space}) is composed by
the vectors $\ket{i,w}$ where $i\in\left\{ 0,\ldots,n\right\}$ and $w$
is a possible state
of the working memory. Thus, we can define the Hilbert spaces
associated to each register: (i) the \textit{query space} $H_{Q}$ is
spanned by vectors $\{ \ket{i} : 0\leq i \leq n \}$; (ii) the
\textit{work space} $H_{W}$ is spanned by vectors $\ket{w}$, where $w$
is in the set of allowed values for the working memory. Hence, $H_{A}
= H_{Q} \otimes H_{W}$. If $\ket{\Psi} \in H_{A}$, then it can be
written uniquely as
\begin{equation}
  \ket{\Psi} =\sum_{i=0}^{n}\ket{i} \ket{\Psi_{i}}, 
\end{equation} where $\ket{\Psi_{i}} \in H_{W}.$
The oracle operator $O_{x}$ for some input $x\in\left\{ 0,1\right\} ^{n}$ is defined as
\begin{equation}
\label{oper}
  O_{x}\ket{i} \ket{\Psi_{i}} =\left(-1\right)^{x_{i}}\ket{i} \ket{\Psi_{i}},
\end{equation} where query space has dimension $n+1$ and $x_{0}=0$ is not considered part of the input.  
In our setting it is very important to define $x_{0}=0$, otherwise we could not compute a wide range of functions. However, it can be avoided in other equivalent descriptions of QQM, by using a slightly different definition of the oracle operator~\cite{KAYE}. 

A quantum query algorithm with an output domain $T$ is determined by:
(i) the number $b$ of qubits in the working memory; (ii) a sequence of
unitary operators $\left\{ U_{i}:0\leq i\leq t\right\}$ in $H_{A}$;
and (iii) a CSOP over $H_{A}$ with projectors indexed by elements of
$T$ for the final measurement.

The execution of the algorithm for input $x$ produces a final
state \begin{equation}\left|\Psi^{f}_{x}\right\rangle =
  U_{t}O_{x}U_{t-1}\ldots
  U_{1}O_{x}U_{0}\left|0,0\right\rangle.\end{equation} The number of
queries is defined as the number of times $O_{x}$ occurs in the
execution.  The output $z\in T$ is chosen with a probability
$\pi_{x}\left(z\right)=\left\Vert P_{z}\left|\Psi_{x}^{f}\right\rangle
\right\Vert ^{2}$, using the CSOP.

We say that an algorithm computes a function $f:\left\{ 0,1\right\}
^{n}\rightarrow T$ within error $\varepsilon$ if for all input $x$,
there is $\pi_{x}\left(f\left(x\right)\right) \geq 1-\varepsilon$.
An algorithm is \textit{exact} if $\varepsilon=0$ and is
\textit{bounded-error} if $\varepsilon \leq 1/3$.  

\section{A reformulation of the Quantum Query Model}
\label{S2}

First, we need to introduce a
sequence of unitary operators
$$\widetilde{U}_{k}=U_{k}U_{k-1}\ldots U_{0}$$ for $t\geq k$. Let $\left\{
  P_{k}:0\leq k\leq n\right\}$ be a CSOP where each $P_{i}$ has as
range the subspace of vectors of the form~$\left|i\right\rangle
\left|\psi\right\rangle$, where $\left|i\right\rangle \in H_{Q}$ and
$\left|\psi\right\rangle \in H_{W}$. Finally, we introduce the notation
$\widetilde{P}_{i}^{j}=\widetilde{U}_{j}^{\dagger}P_{i}\widetilde{U}_{j}$. From
Lemma~\ref{lem1} we know that $\left\{ \widetilde{P}_{k}^{j}:0\leq
  k\leq n\right\}$ is also a CSOP for any fixed $j$.

We denote an \textit{algorithm} (without the measurement step) by the $7$-tuple
\[ \mathcal{A} = \left(t,n,m,H_{Q},H_{W},\Psi,\left\{ U_{i}\right\}
\right), \] where $\dim\left(H_{Q}\right)=n+1$,
$\dim\left(H_{W}\right)=m$ and $\ket{ \Psi } \in H_{A}$ is a unit
vector and the unitary operators in $\left\{ U_i : 0 \leq i \leq t+1
\right\}$ are defined on $H_{A}$. In the present work, we always
consider algorithms according to the above definition, unless otherwise stated.

This is all the information required for describing an algorithm using
$t+1$ queries and initial state $\left|\Psi\right\rangle $.
Choosing an arbitrary initial state $\left|\Psi\right\rangle$ is the
same as using $\left|0,0\right\rangle$, but we change the convention
because is more convenient for upcoming notations.

\begin{definition}
  \label{debs}
  Let $a=\left(a_{0},a_{1},\ldots,a_{t}\right)$ be a vector and let
  $\mathbb{Z}_{n+1}=\left\{ 0,1,\ldots,n\right\}$. We say that an
  indexed set of vectors $\left\{
    \left|\Psi\left(k\right)\right\rangle \right.$$\in
  H_{A}:$$\left.k\in \mathbb{Z}_{n+1}^{t+1} \right\}$ \emph{is
    associated with}
  \begin{equation*}\mathcal{A}=\left(t,n,m,H_{Q},H_{W},\Psi,\left\{
        U_{i}\right\} \right)\end{equation*} if we have that
  \begin{equation}\label{ddesc}
    \left|\Psi\left(a\right)\right\rangle =\widetilde{P}_{a_{t}}^{t}\ldots\widetilde{P}_{a_{1}}^{1}\widetilde{P}_{a_{0}}^{0}\left|\Psi\right\rangle,
  \end{equation} where $0\leq a_{i}\leq n$ for all $i$.
\end{definition}

The motivation for this definition is better understood by
considering the following equation. 
Notice that, by using the operators defined above, we have
the 
expression
\begin{align}
  \label{decc}
  \left|\Psi\right\rangle &=\left(\sum_{k_{t}=0}^{n}\widetilde{P}_{k_{t}}^{t}\right)\ldots\left(\sum_{k_{0}=0}^{n}\widetilde{P}_{k_{0}}^{0}\right)\left|\Psi\right\rangle\nonumber\\ &=\sum_{k_{t}=0}^{n}\ldots\sum_{k_{0}=0}^{n}\left|\Psi\left(k_{0},\ldots,k_{t}\right)\right\rangle. 
\end{align}
for any vector.
If $\left|\Psi\right\rangle$ is a unit vector then the set $\left\{
  \left|\Psi\left(k_{0},\ldots,k_{t}\right)\right\rangle \right\}$ is
associated with some algorithm $\mathcal{A}$, whose initial state is
$\left|\Psi\right\rangle$ and has $t+1$ queries. 
Thus, we can interpret
Eq.~(\ref{decc}) as a decomposition of an initial state. 

This
decomposition has an important property that will be given by Theorem~\ref{desc1}.
However, we need to introduce a useful identity first.  Since $$\left\{
  \widetilde{P}_{i}^{j}:0\leq i\leq n\right\}$$ is a CSOP for each $j$
and
\begin{equation}
  \label{eq}
  O_{x}\left|\Psi\right\rangle =\sum_{i\in\left\{ k:x_{k}=0\right\} }P_{i}\left|\Psi\right\rangle -\sum_{i\in\left\{ k:x_{k}=1\right\} }P_{i}\left|\Psi\right\rangle,
\end{equation} 
we get
\begin{equation}
  \label{eqx}
  \widetilde{U}_{j}^{\dagger}O_{x}\widetilde{U}_{j}\left|\Psi\right\rangle =\sum_{i\in\left\{ k:x_{k}=0\right\} }\widetilde{U}_{j}^{\dagger}P_{i}\widetilde{U}_{j}\left|\Psi\right\rangle -\sum_{i\in\left\{ k:x_{k}=1\right\} }\widetilde{U}_{j}^{\dagger}P_{i}\widetilde{U}_{j}\left|\Psi\right\rangle.
\end{equation}

\begin{theorem}
\label{desc1}
If 
the indexed vector $\left\{ \left|\Psi\left(k\right)\right\rangle \in
  H_{A}:k\in \mathbb{Z}_{n+1}^{t+1} \right\}$ is associated with the
algorithm $\mathcal{A}=\left(t,n,m,H_{Q},H_{W},\Psi,\left\{
    U_{i}\right\} \right),$ then
\begin{equation}
  \label{desc}
\widetilde{U}_{t}^{\dagger}O_{x}U_{t}\ldots U_{1}O_{x}U_{0}\left|\Psi\right\rangle =\sum_{k_{t}=0}^{n}\ldots\sum_{k_{0}=0}^{n}(-1)^{\sum_{i=0}^{t} x_{k_{i}} }\left|\Psi\left(k_{0},\ldots,k_{t}\right)\right\rangle.
\end{equation}
\end{theorem}
\begin{proof}
We shall prove by induction on $t$.  First, as we state our induction
hypothesis, notice that Eq.~\eqref{desc} holds when $t=0$:
\begin{align}
  \widetilde{U}_0^\dagger O_x U_0 \ket{\Psi} &= U_0^\dagger O_x U_0 \ket{ \Psi } \nonumber\\
  &= \sum_{k_0 = 0}^n
  (-1)^{ x_{k_{0}} }\left|\Psi\left(k_{0}\right)\right\rangle.
\end{align}

Then, notice that if Eq.~\eqref{desc} holds for a particular $t$, then
it must hold for $t+1$. That is, if
\[\widetilde{U}_{t}^{\dagger}O_{x} \ldots
O_{x}U_{0}\left|\Psi\right\rangle =\sum_{k_{t}=0}^{n} \ldots
\sum_{k_{0}=0}^{n}(-1)^{\sum_{i=0}^{t} x_{k_{i}} }\left|\Psi\left(k_{0},
    \ldots ,k_{t}\right)\right\rangle\] then, using the
Eq.~\eqref{eqx}, we have
that 
\begin{multline*}
  \widetilde{U}_{t+1}^\dagger O_x \ldots O_x U_0 \ket{\Psi} =
  \sum_{k_{t}=0}^{n} \ldots
  \sum_{k_{0}=0}^{n}(-1)^{\sum_{i=0}^{t} x_{k_{i}} }
  \sum_{k_{t+1}=0}^{n}(-1)^{x_{k_{t+1}}}\widetilde{P}_{k_{t+1}}^{t+1}\left|\Psi\left(k_{0},
      \ldots ,k_{t}\right)\right\rangle.
\end{multline*}
Reordering the summations, we have
\[\widetilde{U}_{t+1}^{\dagger}O_{x}\ldots
O_{x}U_{0}\left|\Psi\right\rangle=\sum_{k_{t+1}=0}^{n} \ldots
\sum_{k_{0}=0}^{n}(-1)^{\sum_{i=0}^{t+1} x_{k_{i}} }\widetilde{P}_{k_{t+1}}^{t+1}\left|\Psi\left(k_{0},\ldots,k_{t}\right)\right\rangle.\]
According to Definition~\ref{debs} and observing the notation for $\widetilde{P}_{i}^{j}$, we finally
have 
\[\widetilde{U}_{t+1}^{\dagger}O_{x} \ldots
O_{x}U_{0}\left|\Psi\right\rangle=\sum_{k_{t+1}=0}^{n} \ldots
\sum_{k_{0}=0}^{n}(-1)^{\sum_{i=0}^{t+1} x_{k_{i}} }\left|\Psi\left(k_{0},
    \ldots ,k_{t+1}\right)\right\rangle.\]
\end{proof}

\begin{corollary} \label{desc2} Consider the vectors
  $\left|\bar{\Psi}\left(k_{0}, \ldots ,k_{t}\right)\right\rangle
  =U_{t+1}\widetilde{U}_{t}\left|\Psi\left(k_{0}, \ldots
      ,k_{t}\right)\right\rangle$ $\forall k_{i}\in\mathbb{Z}_{n+1}$.
  If $\left\{ \left|\Psi\left(k\right)\right\rangle \in H_{A}:k\in
    \mathbb{Z}_{n+1}^{t+1}\right\}$ is associated with $$
  \mathcal{A}=\left(t,n,m,H_{Q},H_{W},\Psi,\left\{ U_{i}\right\}
  \right)$$ then
  \begin{equation}
    U_{t+1}O_{x}U_{t}\ldots U_{1}O_{x}U_{0}\left|\Psi\right\rangle =\sum_{k_{t}=0}^{n}\ldots\sum_{k_{0}=0}^{n}\left(-1\right)^{\sum_{i=0}^{t} x_{k_{i}} }\left|\bar{\Psi}\left(k_{0},\ldots,k_{t}\right)\right\rangle.
  \end{equation}
\end{corollary}

Corollary~\ref{desc2} shows that any quantum algorithm can be
represented as a sum of invariant vectors, whose signs are changed
depending on the input. In 
Fig.~\ref{im2}, we can see an example by means of a graphical representation.

Notice that the algorithm of
Corollary~\ref{desc2} is equivalent to the algorithm of
Theorem~\ref{desc1}, because both algorithms have the same Gram matrices for the final
states $\left|\Psi_{x}^{f}\right\rangle$. So we can ignore the last
unitary operator $U_{t+1}$ and conclude that the algorithm is
determined by the vectors that appear in the decomposition, see
Eq.~(\ref{decc}). The following definition and theorems show us that
we can reformulate the QQM by using this decomposition instead of
unitary operators.

\begin{definition}[Block Set]
  \label{block}
  Let $n,t\geq0$. We say that an indexed set $\left\{
    \left|\Psi\left(k\right)\right\rangle \right.$$\in H_{1}\otimes
  H_{2}:$$\left.k\in \mathbb{Z}_{n+1}^{t+1} \right\}$ is a \emph{Block
    Set} for the ordered pair of Hilbert spaces
  $\left(H_{1},H_{2}\right)$, if:

  \begin{itemize}
  \item $\left\langle \Psi^{i}\left(b_{0}, \ldots
        ,b_{t-i}\right)\mid\Psi^{i}\left(c_{0}, \ldots
        ,c_{t-i}\right)\right\rangle=0$ if $b_{t-i}\neq c_{t-i}$ for
    $0\leq i\leq t$, where 
    the vector $\left|\Psi^{i}\left(a_{0}, \ldots
        ,a_{t-i}\right)\right\rangle$ is defined
    as $$\underset{k_{1}=0}{\overset{n}{\sum}} \ldots
    \underset{k_{i}=0}{\overset{n}{\sum}}\left|\Psi\left(a_{0}, \ldots
        ,a_{t-i},k_{1}, \ldots ,k_{i}\right)\right\rangle;$$
  \item $\underset{k_{0}=0}{\overset{n}{\sum}} \ldots
    \underset{k_{t}=0}{\overset{n}{\sum}}\left\Vert
      \left|\Psi\left(k_{0}, \ldots ,k_{t}\right)\right\rangle
    \right\Vert ^{2}=1;$
   
  \item $\dim(H\left(i,j\right))\leq\dim(H_{2})$ for all $i,j$, where
    $H\left(i,j\right)$ is the Hilbert space spanned
    by vectors $\left\{ \left|\Psi^{t-i}\left(a_{0}, \ldots
          ,a_{i-1},j\right)\right\rangle : a_{k}\in\mathbb{Z}_{n+1}
    \right\};$ and,
  \item $n=\dim(H_1)-1.$
  \end{itemize}
\end{definition}

\begin{theorem}
\label{blothe}
If an indexed set of vectors $\left\{
  \left|\Psi\left(k\right)\right\rangle \in H_{A}:k\in
  \mathbb{Z}_{n+1}^{t+1}\right\}$ is \emph{associated} with an
algorithm $\mathcal{A}=\left(t,n,m,H_{Q},H_{W},\Psi,\left\{
    U_{i}\right\} \right),$ then it is a Block Set for
$\left(H_{Q},H_{W}\right)$.

\end{theorem} 
\begin{proof}
We divide the proof into four parts, 
each corresponding
to one of the four properties 
from the definition of a Block Set:
\begin{enumerate}
\item Since $\{ \widetilde{P}_{k}^{t-i}:0\leq k\leq n\}$ is a
  CSOP, and using the fact that
  $$\left|\Psi^{i}\left(a_{0},a_{1},\ldots,a_{t-i}\right)\right\rangle
  =\widetilde{P}_{a_{t-i}}^{t-i}\ldots\widetilde{P}_{a_{1}}^{1}\widetilde{P}_{a_{0}}^{0}\left|\Psi\right\rangle,$$
  then, whenever $b_{t-i}\neq c_{t-i}$, we have $\left\langle
    \Psi^{i}\left(b_{1},\ldots,b_{t-i}\right)\mid\Psi^{i}\left(c_{1},\ldots,c_{t-i}\right)\right\rangle
  =0$.
\item The second property is proved by using CSOP properties and mathematical
  induction.
\item The space generated by each $\{
    \left|\Psi^{i}\left(a_{0},\ldots,a_{t-i-1},j\right)\right\rangle
    :a_{k}\in\mathbb{Z}_{n+1} \}$ is the same space generated by
  $\{\widetilde{P}_{j}^{t-i}\widetilde{P}_{a_{t-i-1}}^{t-i-1}\ldots\widetilde{P}_{a_{0}}^{0}\left|\Psi\right\rangle
    :a_{k}\in\mathbb{Z}_{n+1} \}$, which is a subspace of the
  space generated by $\{
    \widetilde{U}_{t-i}^{\dagger}\left|j\right\rangle
    \left|w\right\rangle :w\in H_{W} \} $, with dimension
  $\dim\left(H_{W}\right)$.
\item Finally, the fourth property follows directly from $\dim\left(H_{Q}\right)=n+1$. 
\end{enumerate}
\end{proof}

\begin{theorem}
\label{refff}
Let $\left\{ \left|\Psi\left(k\right)\right\rangle \in H_{A}:k\in
  \mathbb{Z}_{n+1}^{t+1}\right\}$ be a Block Set for
$\left(H_{Q},H_{W}\right)$, then it is associated with some $
\mathcal{A}=\left(t,n,m,H_{Q},H_{W},\Psi,\left\{ U_{i}\right\}
\right)$.

\end{theorem} 
\begin{proof}
First, notice that $t$ and $n$ are trivially obtained from
$$\left\{ \ket{\Psi (k)} \in H_{Q} \otimes H_{W} : k_i \in \mathbb{Z}_{n+1} \right\}$$ and $m$ can be trivially obtained
from $H_{W}$. We still have to obtain the other elements.
For the initial state, we take 
$$\ket{\Psi}=\underset{k_{0}=0}{\overset{n}{\sum}}\ldots\underset{k_{t}=0}{\overset{n}{\sum}}\left|\Psi\left(k_{0},\ldots,k_{t}\right)\right\rangle.$$

Now, we must prove that $\left|\Psi\right\rangle$ is a unit vector. By using
$$\left|\Psi^{i}\left(a_{0},\ldots,a_{t-i}\right)\right\rangle
=\underset{j=0}{\overset{n}{\sum}}\left|\Psi^{i-1}\left(a_{0},\ldots,a_{t-i},j\right)\right\rangle$$
as well as the first item of the Block Set definition, we get
$$\left\Vert
  \left|\Psi^{i}\left(a_{0},\ldots,a_{t-i}\right)\right\rangle
\right\Vert ^{2}=\underset{j=0}{\overset{n}{\sum}}\left\Vert
  \left|\Psi^{i-1}\left(a_{0},\ldots,a_{t-i},j\right)\right\rangle
\right\Vert ^{2}.$$ 
Applying the previous equality recursively in
$\left|\Psi\right\rangle$ and using the second item of the Block Set
definition, we finally get $$\left\Vert \left|\Psi\right\rangle
\right\Vert
^{2}=\underset{k_{0}=0}{\overset{n}{\sum}}\ldots\underset{k_{t}=0}{\overset{n}{\sum}}\left\Vert
  \left|\Psi\left(k_{0},\ldots,k_{t}\right)\right\rangle \right\Vert
^{2}=1.$$
  
In the third part of the proof, we have to construct the unitary
operators of $\mathcal{A}$. Those operators are obtained by using a
construction of the CSOP sequence satisfying Eq.~(\ref{ddesc}) with
the Block Set $\left\{
  \left|\Psi\left(k\right)\right\rangle\right\}$. We define
$H_{1}^{i}=\underset{j}{\bigoplus}H\left(i,j\right)$ and an orthogonal
space $H_{2}^{i}$, such that $H_{A}=H_{1}^{i}\oplus H_{2}^{i}$. 
We have
$\dim(H(i,j)) \leq \dim\left(H_{W}\right)$
from the third property of Definition~\ref{block}. 
If
$B\left(i\right)$ is an orthogonal basis of $H_{2}^{i}$, then for each
pair $(i,j)$ we can take
$\dim\left(H_{W}\right)- \dim(H(i,j))$
linearly independent elements from $B\left(i\right)$ and
write such set as $B^{i}_{j}$. The space generated by $B^{i}_{j}$ is
represented as $\widehat{H}\left(i,j\right)$. We also define a space
$\widetilde{H}\left(i,j\right)=\widehat{H}\left(i,j\right)\oplus
H\left(i,j\right)$ with the same dimension of $H_{W}$. We take
the Hilbert spaces $\widehat{H}\left(i,j\right)$, which are pairwise orthogonal for different
$j$. This is possible because $0\leq j\leq n$ and
$\dim\left(H_{A}\right)=\left(n+1\right)\dim\left(H_{W}\right)$. Thereby
$j_{1}\neq j_{2}$ implies that $\widetilde{H}\left(i,j_{1}\right)$ and
$\widetilde{H}\left(i,j_{2}\right)$ are orthogonal. Then, for each $i$
there is a CSOP $\left\{ \widetilde{P}_{j}^{i}:0\leq j\leq n\right\}$
such that $\widetilde{H}\left(i,j\right)$ is the range of the
projector $\widetilde{P}_{j}^{i}$. From Lemma~\ref{lem2}, there is a
unitary operator $\widetilde{U}_{i}$ such that
$\widetilde{U}_{i}^{\dagger}P_{j}\widetilde{U}_{i}=\widetilde{P}_{j}^{i}$,
as the CSOP $\left\{ P_{k}:0\leq k\leq n\right\}$ was defined. Thus, we
obtain the unitary operators from $U_{0}=\widetilde{U}_{0}$ and
$U_{i}=\widetilde{U}_{i}\widetilde{U}_{i-1}^{\dagger}$ for
$i>0$.
\end{proof}

Thus, we can say that for any algorithm  there is a Block Set, and for any
Block Set there is an algorithm. The reformulation is almost complete,
except for one question: while an algorithm is associated to a unique Block Set, one
Block Set may be associated to multiple algorithms.  
The 
following theorem implies that a non-bijective relation between both models
is not a problem.

\begin{theorem}
If two different algorithms are associated to the same Block Set
$$\left\{ \left|\Psi\left(k\right)\right\rangle \in
H_{A}: k\in \mathbb{Z}_{n+1}^{t+1} \right\},$$ then they have
the same Gram matrices for their final states. 
\end{theorem}
\begin{proof}
As it was defined, a set $\left\{
  \left|\Psi\left(k\right)\right\rangle \in H_{A}:k\in
  \mathbb{Z}_{n+1}^{t+1}\right\}$ associated to an algorithm just
depends on the unitary operators before the last query. Suppose that
two algorithms are associated to the same set. From Corollary~\ref{desc2},
we have that the final state of the algorithms is equal to
the same linear combination of elements from the Block Set for a fixed input $x$, but
is different in the unitary operators applied over each sum. Then,
$\left\langle \Psi_{x}^{f}\right|\left.\Psi_{y}^{f}\right\rangle$ are
equal in both algorithms.
\end{proof}
\begin{definition}
  \label{decal}
  The output state of the input $x$ under a Block Set $$\left\{
    \left|\Psi\left(k\right)\right\rangle \in H_{A}:k\in
    \mathbb{Z}_{n+1}^{t+1}\right\}$$ is defined as
\begin{equation}
  \left|\Psi_{x}^{f}\right\rangle =\sum_{k_{t}=0}^{n}\ldots\sum_{k_{0}=0}^{n}\left(-1\right)^{\sum_{i=0}^{t} x_{k_{i}} }\left|\Psi\left(k_{0},\ldots,k_{t}\right)\right\rangle.
\end{equation}

\end{definition}
This definition closes our new formulation, by defining the
output states from the Block Set that describes the algorithm. Notice
that the space $H_{A}$ is maintained and the Gram matrix of final
states from such Block Set is equal to the Gram matrix for final
states of any algorithm associated to the Block Set. If we keep the same
  measurement step as in the original model, then we can compute the same functions
within the same margin of error in the associated BSF, as we would with the QQM
algorithms. In fact, it is just a matter of choosing adequate measurement
steps. In Appendix~A we present an example of a Block Set equivalent to a QQM
algorithm.

\section{Gram matrices and Block Sets}
\label{secpbs}

At this point, Block Sets are taken as an equivalent parametrization
of quantum query algorithms, where we consider the elements of a Block
Set as the new parameters. In this section, we study how each element
will affect the final Gram matrix of output states. That information
can open the possibility of using such parameters for constructing a
Gram matrix, that is appropriate for computing a given function.  If
inputs $x$ and $y$ should give different outputs for a given function,
then the quantum algorithm must be designed for making
$\left\langle \Psi_{x}^{f}|\Psi_{y}^{f}\right\rangle$ as close to zero
as possible.

It is convenient to introduce four auxiliary vectors, as
  follows.  Vector $\ket{A}$ is defined as the sum of those components
  $\ket{\Psi(a)}$ of a Block Set whose sign is kept unchanged in both
  $\ket{\Psi_x^f}$ and $\ket{\Psi_y^f}$.  Analogously, vector
  $\ket{B}$ is defined as the sum of those components $\ket{\Psi(a)}$ of a Block
  Set whose sign is kept unchanged in $\ket{\Psi_x^f}$ while inverted
  in $\ket{\Psi_y^f}$.  Vector $\ket{C}$ is defined as the sum of those
  components $\ket{\Psi(a)}$ of a Block Set whose sign is inverted in
  both $\ket{\Psi_x^f}$ and $\ket{\Psi_y^f}$.  Finally, vector
  $\ket{D}$ is defined as the sum of those components $\ket{\Psi(a)}$ of a Block
  Set whose sign is inverted in $\ket{\Psi_x^f}$ while kept unchanged
  in $\ket{\Psi_y^f}$.
Notice that
$\left|\Psi\right\rangle =\left|A\right\rangle +\left|B\right\rangle
+\left|C\right\rangle +\left|D\right\rangle$,
and
$\left|\Psi_{x}^{f}\right\rangle=\left|A\right\rangle
+\left|B\right\rangle -\left|C\right\rangle -\left|D\right\rangle$,
and
$\left|\Psi_{y}^{f}\right\rangle=\left|A\right\rangle
-\left|B\right\rangle -\left|C\right\rangle +\left|D\right\rangle$.

Expanding $\left\langle \Psi_{x}^{f}|\Psi_{y}^{f}\right\rangle$ and
$\left\langle \Psi|\Psi\right\rangle$ in terms of the above defined vectors and
summing those expressions, we get
\begin{equation}
  \label{EQU8}
  \left\langle \Psi_{x}^{f}|\Psi_{y}^{f}\right\rangle =2\left(\left\langle \Psi_{x}^{+}|\Psi_{y}^{+}\right\rangle +\left\langle \Psi_{x}^{-}|\Psi_{y}^{-}\right\rangle \right)-1,
\end{equation} where  $\left|\Psi_{x}^{\pm}\right\rangle =\frac{\pm\left|\Psi_{x}^{f}\right\rangle +\left|\Psi\right\rangle }{2}$, and analogously for $\left|\Psi_{y}^{\pm}\right\rangle$.

We say that a Block Set is real-valued if its elements are vectors on the
real numbers. 

\begin{lemma}
  \label{lemma3}
  If there is a complex Block Set
  $\left\{ \left|\Psi\left(k\right)\right\rangle \in H_{A}:k\in
    \mathbb{Z}_{n+1}^{t+1}\right\}$
  for $\left(H_{Q},H_{W}\right)$, whose output states are used for
  computing a function
  $f:\left\{ 0,1\right\} ^{n}\rightarrow\left\{ 0,1\right\}$ within
  error $\epsilon$, then there is a real Block Set
  $\left\{ \left|\widehat{\Psi}\left(k\right)\right\rangle
  \right.$$\in\widehat{H}_{Q}\otimes\widehat{H}_{W}:$$\left.k\in
    \mathbb{Z}_{n+1}^{t+1} \right\} $
  for some $\left(\widehat{H}_{Q},\widehat{H}_{W}\right)$, whose
  output spaces can be used to compute $f$ within the same error.
\end{lemma}

\begin{proof}
If the outputs from
$\left\{ \left|\Psi\left(k\right)\right\rangle \right\}$ can be used
for computing $f$ within error $\epsilon$ (with an appropriate CSOP),
then the existence of a quantum query algorithm that computes $f$ within error
$\epsilon$ in $t+1$ queries follows directly from Theorems~\ref{desc1}
and \ref{refff}.
Barnum et al. \cite{BARNUM} proved that
there exists a quantum algorithm that computes $f$ within error
$\epsilon$ in $t+1$ queries, if and only if a semi-definite program
$P\left(f,t+1,\epsilon\right)$ is feasible, where the unitary matrices and states in the quantum query algorithm corresponding to a solution for $P\left(f,t+1,\epsilon\right)$ can be taken to be real; Montanaro, Jozsa and Mitchison~[5] gave an explicit construction achieving this.
The set of vectors
$\left\{ \left|\widehat{\Psi}\left(k\right)\right\rangle \right\}$
associated to this algorithm has output states that produces the same
Gram matrix by Theorem~\ref{desc1}   
and all its elements are real. Finally, this set of vectors is a Block
Set according to Theorem~\ref{blothe}.
\end{proof}

According to Lemma~\ref{lemma3}, we can always assume that Block Sets are real-valued,
without loss of generality. The 
following lemma presents a useful property about this particular
case.

\begin{lemma}
  \label{prolem}
  If a Block Set
  $\left\{ \left|\Psi\left(k\right)\right\rangle \in H_{A}:k\in
    \mathbb{Z}_{n+1}^{t+1}\right\}$
  is real, then for any input $x\in\left\{ 0,1\right\} ^{n}$,
  $\left\langle \Psi_{x}^{+}\mid\Psi_{x}^{-}\right\rangle =0$.
\end{lemma}

\begin{proof}
There is
$\left\langle \Psi_{x}^{+}\mid\Psi_{x}^{-}\right\rangle
=\frac{1}{4}\left(\left\Vert \left|\Psi\right\rangle \right\Vert
  ^{2}+\left\langle \Psi_{x}^{f}\mid\Psi\right\rangle -\left\langle
    \Psi\mid\Psi_{x}^{f}\right\rangle -\left\Vert
    \left|\Psi_{x}^{f}\right\rangle \right\Vert ^{2}\right)$.

If the Block Set is real, then $\left|\Psi\right\rangle$ and
$\left|\Psi_{x}^{f}\right\rangle$ are real unit vectors, then
$\left\langle \Psi_{x}^{f}\mid\Psi\right\rangle =\left\langle
  \Psi\mid\Psi_{x}^{f}\right\rangle$,
in addition
$\left\Vert \left|\Psi\right\rangle \right\Vert =\left\Vert
  \left|\Psi_{x}^{f}\right\rangle \right\Vert =1$
implies $\left\langle \Psi_{x}^{+}\mid\Psi_{x}^{-}\right\rangle
=0$.
\end{proof}
\begin{theorem}
\label{prothe}
Let the vectors $\ket{A},\ \ket{B},\ \ket{C},$ and $\ket{D}$ be as they were defined, however with the
additional condition of being real-valued. Then we get
\begin{equation}
  \label{EQPRO}
  \left\langle \Psi_{x}^{f}\mid\Psi_{y}^{f}\right\rangle =2\left(\left\Vert \left|A\right\rangle \right\Vert ^{2}-2\left\langle A\mid C\right\rangle +\left\Vert \left|C\right\rangle \right\Vert ^{2}\right)-1.
\end{equation}
\end{theorem}
\begin{proof}
Using Lemma~\ref{prolem} over
$\left(\left|\Psi_{x}^{+}\right\rangle
  ,\left|\Psi_{x}^{-}\right\rangle \right)$
and
$\left(\left|\Psi_{y}^{+}\right\rangle
  ,\left|\Psi_{y}^{-}\right\rangle \right)$,
there are two equations. We consider a system of equations, joining
the two last equations with Eq.~(\ref{EQU8}). Expressing that system
in dot products of $\ket{A}$, $\ket{B}$, $\ket{C}$ and $\ket{D}$, we obtain a new system that
derives Eq.~(\ref{EQPRO}), by elementary algebra.
\end{proof}

The previous theorem give us a way of obtaining the Gram matrix of final
states directly from a given Block Set.

Let
$\mathcal{B}=\left\{ \left|\Psi\left(k\right)\right\rangle \in
  H_{A}:k\in \mathbb{Z}_{n+1}^{t+1}\right\}$
be a Block Set for
$\left(H_{Q},H_{W}\right)$.  
We denote \begin{equation*} k=\left(k_{0},k_{1},\ldots,k_{t}\right)
\end{equation*} and the following
subsets of $\mathcal{B}$:
\begin{enumerate}
\item $\mathcal{B}_{x}^{+}=$$\left\{
    \left|\Psi\left(k\right)\right\rangle \in
    H_{A}:\left(-1\right)^{\sum_{i=0}^{t} x_{k_{i}} }=1\right\}
  $.
\item
  $\mathcal{B}_{x}^{-}=$$\left\{
    \left|\Psi\left(k\right)\right\rangle \in
    H_{A}:\left(-1\right)^{\sum_{i=0}^{t} x_{k_{i}} }=-1\right\}
  $.
\end{enumerate}
Then $\widetilde{A}_{xy}=\mathcal{B}_{x}^{+}\cap \mathcal{B}_{y}^{+}$
and $\widetilde{C}_{xy}=\mathcal{B}_{x}^{-}\cap \mathcal{B}_{y}^{-}$.

Notice that $\mathcal{B}_{x}^{+}$ and $\mathcal{B}_{x}^{-}$ are the
sets of positive and negative terms in Eq.~(\ref{desc}),
respectively. So for each pair $x,y$, the sets $\widetilde{A}_{xy}$ and
$\widetilde{C}_{xy}$ contain vectors of a Block Set, whose sum define
$\left|A\right\rangle$ and~$\left|C\right\rangle $, respectively.

\begin{lemma}
  \label{thethe}
  Let
  $\mathcal{B}=\left\{ \left|\Psi\left(k\right)\right\rangle \in
    H_{A}:k\in \mathbb{Z}_{n+1}^{t+1}\right\}$
  be a Block Set for $\left(H_{Q},H_{W}\right)$, where:
  \begin{itemize}
  \item $P\left(k\right)$ is the set of pairs $\left(x,y\right)$ such
    that
    $\left|\Psi\left(k\right)\right\rangle \in\widetilde{A}_{xy}$.
  \item $Q\left(k\right)$ is the set of pairs $\left(x,y\right)$ such
    that
    $\left|\Psi\left(k\right)\right\rangle \in\widetilde{C}_{xy}$.
  \end{itemize}
  Then
  $P\left(k\right)=\left\{ x:\left(x_{k_{0}}\oplus\ldots\oplus
      x_{k_{t}}\right)=0\right\} ^{2}$
  and
  $Q\left(k\right)=\left\{ x:\left(x_{k_{0}}\oplus\ldots\oplus
      x_{k_{t}}\right)=1\right\} ^{2}$.
\end{lemma}

\begin{proof}
Using the definitions of $\widetilde{A}_{xy}$ and
$\widetilde{C}_{xy}$, we have
\begin{equation}
  P\left(k\right)=\left\{ \left(x,y\right): (-1)^{\sum_{i=0}^{t} x_{k_{i}} }=1 \mbox{ and } (-1)^{\sum_{i=0}^{t} y_{k_{i}} }=1 \right\}
\end{equation}
and
\begin{equation}
  Q\left(k\right)=\left\{ \left(x,y\right):(-1)^{\sum_{i=0}^{t} x_{k_{i}} }=-1\mbox{ and }(-1)^{\sum_{i=0}^{t} y_{k_{i}} }=-1\right\}.
\end{equation}

Notice that $x$ does not have influence on the predicate of $y$, nor $y$ have influence on
the predicate of $x$. Therefore, the sets of allowed values for $x$ and $y$
form a Cartesian
product. Notice that $x_{k_{0}}\oplus\ldots\oplus
      x_{k_{t}} =0$ iff $(-1)^{\sum_{i=0}^{t} x_{k_{i}} }=1$. 
\end{proof}

Now we may define the square matrices $\bar P_{k,h}$ and $\bar Q_{k,h}$, with row $x$ and column $y$ being indexed by elements
of $\{0,1\}^n$ and with entries taking values in $\{0,1\}$, as follows:
\begin{itemize}
\item $\bar P_{k,h}[x,y]=1$ iff
  $\left(x,y\right)\in P\left(k\right)\cap
    P\left(h\right)$;
\item $\bar Q_{k,h}[x,y]=1$ iff
  $\left(x,y\right)\in Q\left(k\right)\cap
    Q\left(h\right)$;
  and,
\item $\bar R_{k,h}[x,y]=1$ iff
  $\left(x,y\right)\in P\left(k\right)\cap
    Q\left(h\right)$.
\end{itemize}

\begin{theorem}
\label{thema}
Let
$\mathcal{B}=\left\{ \left|\Psi\left(k\right)\right\rangle \in
  H_{A}:k\in \mathbb{Z}_{n+1}^{t+1}\right\}$
be a real Block Set for $\left(H_{Q},H_{W}\right)$, then the Gram
matrix of their output states
$\left\{ \left|\Psi_{x}^{f}\right\rangle \right\}$ is
\begin{equation}\label{eqsum2}
G = 2\underset{k,h}{\sum}\left(\bar P_{k,h}-2\bar R_{k,h}+\bar Q_{k,h}\right)\left\langle \Psi\left(k\right)\mid\Psi\left(h\right)\right\rangle -J,
\end{equation} 
where $J$ is a matrix where every element is equal to one.
\end{theorem}
\begin{proof}
Follows directly from Eq.~(\ref{EQPRO}), by rewriting the matrices
$\left\{ \bar P_{k,h}\right\} $, $\left\{ \bar R_{k,h}\right\} $ and
$\left\{ \bar Q_{k,h}\right\} $.
\end{proof}

This theorem gives an explicit expression on how pairs of elements in
a Block Set control the Gram matrix of output states. We can think of
each matrix $\bar P_{k,h}-2\bar R_{k,h}+\bar Q_{k,h}$ like acting as a
mask over the Gram matrix. Instead of this general case, there is a
simpler case computationally less powerful, however with a simpler
Gram matrix representation.

\begin{definition}
  \label{bortho}
  A Block Set
  $\mathcal{B}=\left\{ \left|\Psi\left(k\right)\right\rangle \in
    H_{A}:k\in \mathbb{Z}_{n+1}^{t+1}\right\}$
  for $\left(H_{Q},H_{W}\right)$ is orthogonal, if all its elements
  are orthogonal.
\end{definition}

\begin{corollary}
  \label{bortho2}
  Let
  $\mathcal{B}=\left\{ \left|\Psi\left(k\right)\right\rangle \in
    H_{A}:k\in \mathbb{Z}_{n+1}^{t+1}\right\}$
  be an orthogonal real Block Set for $\left(H_{Q},H_{W}\right)$, then
  the Gram matrix of their output states
  $\left\{ \left|\Psi_{x}^{f}\right\rangle \right\} $ is
\begin{equation}  \label{bortheq}
G=2 \underset{k}{\sum}\left(\bar P_{k,k}+\bar
      Q_{k,k}\right)\left\Vert \left|\Psi\left(k\right)\right\rangle
    \right\Vert ^{^{2}} -J.
\end{equation}
\end{corollary}

\begin{proof}
Simply applying in Eq.~(\ref{eqsum2}) that
$\left\langle \Psi\left(k\right)\mid\Psi\left(h\right)\right\rangle
=0$ for $k\neq h$ and $ \bar R_{k,k}=0$ for all $k$.
\end{proof}

In Appendix B, we apply the ideas introduced in this
  section, and show explicitly how the Block Set determines the Gram
  matrix of output states through Theorem~\ref{thema}.

\section{Towards a framework for analyzing quantum exact algorithms}
\label{S5}
In this section we introduce the BSF as tool for designing and analyzing exact quantum algorithms, this formulation implies linear systems that can admit analytic solutions. We also give examples of this application.

First, we define the set of unknowns $\left\{ w_{kh}:k,h\in \mathbb{Z}_{n+1}^{t+1} \right\}$
for the set $\mathbb{Z}_{n+1}^{t+1}$.
Let $X,Y\subset\left\{ 0,1\right\} ^{n}$ be two disjoint
sets. From this notation, we may consider some useful equations:
\begin{enumerate}
\item For each $\left(x,y\right)\in X\times Y$ there is an equation
  \begin{equation}
    \label{freq1}
    \underset{k,h\in \mathbb{Z}_{n+1}^{t+1}}{\sum}\left(\bar P_{k,h}\left[x,y\right]-2\bar R_{k,h}\left[x,y\right]+ \bar Q_{k,h}\left[x,y\right]\right)w_{kh}=\frac{1}{2}.
  \end{equation}
\item Let $\mathcal{I}_{i}\left(k'\right)=\bigl\{
  k\in  \mathbb{Z}_{n+1}^{t+1} :$
  $\left.  0\leq j\leq
      i \mbox{ and } k_{j}'=k_{j} \mbox{ for all } j  \right\}$ for each
  $k'\in\left(\mathbb{Z}_{n+1}\right) ^{i+1}$. 
  Thus, for each
  $i\in\mathbb{Z}_{t+1} $ and $k',h'\in\left(\mathbb{Z}_{n+1}\right)
  ^{i+1}$, such that $k_{i}'\neq h_{i}'$, there is an equation
  \begin{equation}
    \label{freq2}
    \underset{k\in \mathcal{I}_{i}\left(k'\right)}{\sum}\left(\underset{h\in \mathcal{I}_{i}\left(h'\right)}{\sum}w_{kh}\right)=0.
  \end{equation}
\item And, finally, there is a constraint
\begin{equation}\label{freq3}\underset{k\in
      \mathbb{Z}_{n+1}^{t+1} }{\sum}w_{kk}=1.\end{equation}
\end{enumerate}
The union of all these equations forms a system, which we denote as $E\left(t,n,X,Y\right)$.

\begin{theorem}
\label{linsys}
Let $f:\left\{ 0,1\right\} ^{n}\rightarrow\left\{0,1\right\} $ be a partial
function such that, if $x\in X$ and $y\in Y$,
then $f\left(x\right) \neq f\left(y\right)$.
Then, $f$ is
computed exactly in $t+1$ queries if and only if
$E\left(t,n,X,Y\right)$ has a real solution for $\left\{ w_{kh}:k,h\in
  \mathbb{Z}_{n+1}^{t+1} \right\}$ such that these values under the
same indices form a positive semi-definite matrix. 
\end{theorem}
\begin{proof}
In the first part of the proof, if $f$ can be computed exactly
within $t+1$ queries by a quantum query algorithm $\mathcal{A}$, then
there is a set $\left\{ \left|\Psi\left(k\right)\right\rangle :k\in
  \mathbb{Z}_{n+1}^{t+1} \right\} $ that is associated to the
algorithm $\mathcal{A}$ and this set is a Block Set according to
Theorem~\ref{blothe}. 
If this Block Set is complex then according to Lemma~\ref{lemma3} there is
another real Block Set $\left\{
  \left|\widehat{\Psi}\left(k\right)\right\rangle :k\in
  \mathbb{Z}_{n+1}^{t+1} \right\} $, whose output states can be used
for computing the same function $f$ exactly.

Take $w_{k_1 k_2}=\left\langle
  \widehat{\Psi}\left(k_{1}\right)\mid\widehat{\Psi}\left(k_{2}\right)\right\rangle
$.  Since $f$ is computed exactly, if $x\in X$ and $y\in Y$, then
$f\left(x\right)\neq f\left(y\right)$ and the output states of
$\mathcal{A}$ must be orthogonal, i.e., $\left\langle
  \widehat{\Psi}_{x}^{f}\mid\widehat{\Psi}_{y}^{f}\right\rangle
=0$. Since $\mathcal{A}$ and the Block Set have the same Gram matrix for
output states, then from Theorem~\ref{thema}
we have that Eq.~(\ref{freq1}) is satisfied for $\left(x,y\right)$.
Eq.~(\ref{freq2}) is just another way of writing the first property of
Definition ~\ref{block}. Eq.~(\ref{freq3}) is another way of writing
the second property of Definition~\ref{block}. Finally, the values assigned for
$\left\{ w_{k_{1}k_{2}}\right\} $ are a positive semi-definite matrix
because it is the Gram matrix of $\left\{
  \widehat{\Psi}\left(k\right)\right\} $.

In the second part of the proof, since the values for $\left\{
  w_{k_{1}k_{2}}\right\} $ form a positive semi-definite matrix then
it is a Gram matrix for a set of vectors $\left\{
  \left|\Psi\left(k\right)\right\rangle :k\in \mathbb{Z}_{n+1}^{t+1}
\right\} $. This set of vectors satisfy the first property of Definition~\ref{block} according to Eq.~(\ref{freq2}) and the second property 2 of Definition~\ref{block}  according to Eq.~(\ref{freq3}). If we define the appropriate spaces
$H_{1}$ and $H_{2}$, then the third and fourth properties of Definition~\ref{block} are satisfied and $\left\{\Psi\left(k\right)\right\}$ is a
Block Set. From Eq.~(\ref{freq1}) and Theorem~ref{thema}
, we have that the sets of output states $\left\{ \left|\Psi_{x}^{f}\right\rangle
  ,x\in X\right\} $ and $\left\{ \left|\Psi_{x}^{f}\right\rangle ,x\in
  Y\right\}$ generate two orthogonal spaces. Therefore, there is a CSOP that
allows us to measure the output exactly. From Theorems~\ref{desc1}
and~ref{refff}
, we can conclude that a quantum query algorithm associated to
$\left\{ \Psi\left(k\right)\right\} $ jointly with the CSOP computes
$f$ exactly in $t+1$ queries.
\end{proof}

System $E\left(t,n,X,Y\right)$ has  an
exponential number of variables, then using this theorem for any
numerical procedure is impractical and the theorem itself is difficult
to use as an analytic tool. Another difficulty is maintaining the
semi-definite property of the solution. Nevertheless there exists the
possibility of taking special cases of this general formulation. For
example, if we assume that some variables are equal to zero,
then we can construct particular families of exact quantum algorithms
more easily. This is the strategy that we use in the following corollary 
for obtaining a more practical tool.

Let the system $\widehat{E}\left(t,n,X,Y\right)$ be the union of 
the following equations:
  \begin{equation}\label{freqq1}
    \underset{k\in \mathbb{Z}_{n+1}^{t+1}
    }{\sum}\left(\bar P_{k,k}\left[x,y\right]+ \bar Q_{k,k}\left[x,y\right]\right)w_{kk}=\frac{1}{2},
  \end{equation}
for each $\left(x,y\right)\in X \times Y$, and 
\begin{equation} \label{freqq2} \underset{k\in
      \mathbb{Z}_{n+1}^{t+1} }{\sum}w_{kk}=1. 
\end{equation}

\begin{corollary}
  \label{ficor}
  Let $f:\left\{ 0,1\right\} ^{n}\rightarrow\left\{ 0,1\right\} $ be a partial
  function where $x\in X$ and $y\in Y$
  implies that $f\left(x\right)\neq f\left(y\right)$. 
  If
  $\widehat{E}\left(t,n,X,Y\right)$ has solution over the non-negative
  real numbers, then there is a quantum query algorithm that
  computes $f$ exactly in $t+1$ queries.
\end{corollary}

\begin{proof}
If the Block Set has the restriction of being orthogonal (see Definition
\ref{bortho}) for computing~$f$, then taking $w_{k_{1}k_{2}}=\left\langle
  \Psi\left(k_{1}\right)\mid\Psi\left(k_{2}\right)\right\rangle$ 
  and $\left(k_{1}\neq k_{2}\right)$
implies that $ w_{k_{1}k_{2}}=0$. 
Then, Eq.~(\ref{freq3}) is the same as
Eq.~(\ref{freqq2}), Eq.~(\ref{freq2}) disappears and as $\bar
R_{k,k}=0$ then Eq.~(\ref{freq1}) becomes Eq.~(\ref{freqq1}). Finally, the matrix formed by elements $w_{k_{1}k_{2}}$ has no negative value in the diagonal
 and has zero in the rest, this guarantees the positive
semi-definite property.
\end{proof}

The orthogonality condition takes off computational power of the algorithms that we can obtain. However, this set of algorithms is still interesting. For example, it contains all exact quantum algorithms that use a single query. The largest possible separation between quantum and randomized query complexities can be obtained by a single query quantum algorithm---which is therefore orthogonal---even though this algorithm is not exact~\cite{Aaronson}.
Corollary \ref{ficor} is a much simpler tool, in the sense that each $k\in
\mathbb{Z}_{n+1}^{t+1}$ has an independent influence to the Gram
matrix. Let $T\left(k\right)$ be the set of pairs $\left(x,y\right)$
such that $$\bar P_{k,k}\left[x,y\right]+\bar
Q_{k,k}\left[x,y\right]=1.$$ 
We can say that the weight of
$T\left(k\right)$ on the Gram matrix is controlled by the
value of $w_{kk}$ and the intersection of those sets determines which
regions of $\left\{ 0,1\right\} ^{n}\times\left\{ 0,1\right\} ^{n}$
satisfy Eq.~(\ref{freqq1}). That is equivalent to saying
that those regions have value $0$ in the Gram matrix, and thus
determines which inputs can be computed exactly for a given
algorithm. However, the amount of weight that we can give to each $k$ is
limited by Eq.~(\ref{freqq2}).  It is also important to notice that
increasing $t$ increases the possible shapes for $T\left(k\right)$ and
enlarges the set of possible Gram matrices that we can
obtain. 
We can even imagine a random
procedure for generating arbitrary exact quantum algorithms. The first step is giving
weights for some set of variables $\left\{ w_{kk}:k\in L\subset
  \mathbb{Z}_{n+1}^{t+1}\right\} $ until the limit imposed by
Eq.~(\ref{freqq2}) is reached,  
the last step is searching interesting
sets $X$ and $Y$ such that
$x \in X$ and $y \in Y$ iff
$$\underset{k}{\sum}\left(\bar P_{k,k}\left[x,y\right]+\bar Q_{k,h}\left[x,y\right]\right)w_{kk}=\frac{1}{2}.$$ 
The design of exact quantum algorithms using Corollary~\ref{ficor} 
can be done by
analyzing the possible multiple intersections between the elements in
set $$\left\{ T\left(k\right):k\in \mathbb{Z}_{n+1}^{t+1} \right\}.$$

There are two useful observations that can be considered if we want to
use Corollary~\ref{ficor}. Let $\overline{x}$ be the bit-wise negation of $x\in\left\{ 0,1\right\}
^{n}$, i.e., $\overline{x}\in\left\{ 0,1\right\}^{n}$ such that
$x_{i}\neq\overline{x}_{i}$ for all i.  It is not difficult to
prove that 

$$\bar P_{k,k}\left[x,y\right]+\bar
Q_{k,k}\left[x,y\right]=\bar
P_{k,k}\left[\overline{x},\overline{y}\right]+\bar
Q_{k,k}\left[\overline{x},\overline{y}\right]$$ 
for all $k$, and as a consequence
all Gram matrix $G$ obtained using the corollary have the
property that
$G\left[x,y\right]=G\left[\overline{x},\overline{y}\right]$. Moreover,
if $p\left(k\right)$ represents all the permutations of $k$, then
$\bar P_{k,k}+\bar Q_{k,k}=\bar P_{k',k'}+\bar Q_{k',k'}$ for all
$k'\in p\left(k\right)$. Thus, assigning random values to the set of
unknowns $W\left(k\right)=\left\{ w_{k'k'}:k'\in
  p\left(k\right)\right\}$, keeps the Gram matrix 
  invariant as long as the sum~$\underset{k'\in W\left(k\right)}{\sum}w_{k'k'}$ remains constant.

\subsection{A generalization of the Deutsch-Jozsa algorithm by means
  of the Block Set Formalism}

We show an example of BSF algorithm obtained by this analysis. We
assume that $n$ is even and $n>2t$. Thereby, we define the set
$\left\{ k_{i}:0<i\leq n\right\} \subset \mathbb{Z}_{n+1}^{t+1}$, such
that
$k_{i}=\left(r\left(i\right),r\left(i+1\right),\ldots,r\left(i+t\right)\right)$,
where $r\left(i\right)=i$ for $i\leq n$ and $r\left(i\right)=(i-n)$
for $i>n$. If we take $w_{k_{i}k_{i}}=\frac{1}{n}$ for all $0<i\leq
n$, then the system $\widehat{E}\left(t,n,X,Y\right)$ is satisfied for
$X=\left\{ 0^{n},1^{n}\right\}$ and $Y=\left\{ x\in\left\{ 0,1\right\}
  ^{n}:S\left(x\right)=\frac{n}{2}\right\}$; where we define
$S\left(x\right)$ as the number of satisfied Boolean clauses
$\phi_{i}=x_{r\left(i\right)}\oplus
x_{r\left(i+1\right)}\oplus\ldots\oplus x_{r\left(i+t\right)}$, such
that $0<i\leq n$. We can claim that $\widehat{E}\left(t,n,X,Y\right)$
is satisfied under the following observations. The equation
  
\begin{equation}\label{freqq3}
  \underset{i}{\sum}\left( \bar P_{k_{i},k_{i}}\left[0^{n},y\right]+\bar Q_{k_{i},k_{i}}\left[0^{n},y\right]\right)w_{k_{i}k_{i}}=\frac{1}{2}
\end{equation} is satisfied only if $\frac{n}{2}$ matrices $\bar P_{k_{i},k_{i}}$ are equal to $1$ in column $y$ and row $0^{n}$, because matrices $\bar Q_{k_{i},k_{i}}$ do not have values $1$ on row $0^{n}$. Last claims imply that $S\left(y\right)=\frac{n}{2}$. Finally, 
since  $S\left(x\right)=\frac{n}{2}$ also implies that $S\left(\overline{x}\right)=S\left(x\right)$, we have that Eq.~(\ref{freqq3}) must hold also for $\overline{y}$. Recall that $G\left[x,y\right]=G\left[\overline{x},\overline{y}\right]$. Therefore,
 \begin{equation}
   \underset{i}{\sum}\left(\bar P_{k_{i},k_{i}}\left[1^{n},y\right]+\bar Q_{k_{i},k_{i}}\left[1^{n},y\right]\right)w_{k_{i}k_{i}}=\frac{1}{2}
 \end{equation} for all $y$ such that $S\left(y\right)=\frac{n}{2}$.
  
 Thus, by Corollary~\ref{ficor} there is an exact quantum algorithm
 which computes two different outputs for $X$ and $Y$. The first two
 cases of $t$ are detailed below:
 \begin{itemize}
 \item For $t=0$, there is a BSF algorithm equivalent to Deutsch-Jozsa
   algorithm \cite{Deutsch}.
 \item For $t=1$, there is a BSF algorithm that discriminates $\left\{
     0^{n},1^{n}\right\}$ from $x$, where  there is a set $S$ such that 
$i\in S$ iff $x_{i}=x_{i+1}$ and $\left|S\right|=\frac{n}{2}$.  This is stated
   by defining the first bit as
   following the last bit. This algorithm can be implemented in the QQM by applying Deutsch-Jozsa algorithm over the state $\underset{i}{\sum}\left(-1\right)^{x_{i}+x_{j}}\left|i\right\rangle $, where $j \equiv i+1 \mod n$, which costs two queries.
 \end{itemize}
  
\subsection{Characterizing the power of orthogonal algorithms}
 System $\widehat{E}\left(t,n,X,Y\right)$ implies a clear and
 straightforward view on how orthogonal BSF algorithms work, thus it
 is interesting in a theoretical sense. In practice however, we can work
 with a smaller system as it is proved below.

\begin{theorem}
\label{sma}
 The system $\widehat{E}\left(t,n,X,Y\right)$ is equivalent to the
 system $\widetilde{E}\left(t,n,X,Y\right)$, which is defined as the
 union of
 following equations:
\begin{equation}
  \underset{k\in \mathbb{Z}_{n+1}^{t+1}}{\sum} \bar P_{k,k}    \left[0^{n}, x \oplus y \right]w_{kk}=\frac{1}{2}.
\end{equation}
for each $\left(x,y\right)\in X \times Y$, and 
 \begin{equation} \underset{k\in
      \mathbb{Z}_{n+1}^{t+1}}{\sum}w_{kk}=1.
 \end{equation}
\end{theorem}
\begin{proof}
Let $x \oplus y \in\left\{ 0,1\right\} ^{n}$ be the bit-wise
xor operation between $x$ and $y$. Consider the identity $\bar
P_{k,k}\left[x,y\right]+\bar Q_{k,k}\left[x,y\right]=\bar
P_{k,k}\left[0^{n},x \oplus y\right]$.
Then 
\begin{equation} \underset{k\in
    \mathbb{Z}_{n+1}^{t+1}}{\sum}\left(\bar P_{k,k}\left[x,y\right]+
    \bar Q_{k,k}\left[x,y\right]\right)w_{kk}= \underset{k\in
    \mathbb{Z}_{n+1}^{t+1} }{\sum} \bar
  P_{k,k}\left[0^{n},x \oplus y\right]w_{kk}.
\end{equation}
\end{proof}

Last theorem implies that system $\widetilde{E}\left(t,n,X,Y\right)$
is equivalent to $\widehat{E}\left(t,n,0^{n},Z\right)$, where there is
defined $Z=\left\{ x \oplus y:\left(x,y\right)\in X\times
  Y\right\}$.  Thereby, if an exact orthogonal BSF algorithm
discriminates $0^{n}$ from $Z$, then it also can be used for
discriminating $X$ from $Y$ with error zero. In the case of orthogonal
BSF algorithms, Theorem~9 allows us to simplify the
algorithm-construction problem, we just need to determine which sets
can be discriminated from $0^{n}$ given a bounded $t$.
Recall that
a permutation on vector $k$ gives the same variable $w_{kk}$, besides repeated values $k_i=k_j$ in $k$ also implies redundancy, then system $\widetilde{E}\left(t,n,0^{n},Z\right)$ implies a  matrix of size
$\mathcal{O}\left(2^{n}\right)\times\mathcal{O}\left(2^{n}\right)$. We can compare it with the system given by Barnum, Saks and Szegedy~\cite{BARNUM} which implies $\mathcal{O}\left(t\right)$ matrices of size $\mathcal{O}\left(2^n\right)$, thus $\widetilde{E}\left(t,n,0^{n},Z\right)$ is less powerful but also computationally cheaper.

Corollary~\ref{sinte} characterizes the computational power of orthogonal exact algorithms, but first we define a problem that is general enough for describing any function whose domain is in the hypercube.

\begin{definition}[XOR-Weighted-Problem]

Let be a set of Boolean formulas 
$$\mathbb{X}=\left\{
  \underset{i}{\bigoplus}x_{k_{i}}:x_{0}=0,k\in
  K\subset\mathbb{Z}_{n+1}^{t+1}\right\},$$ 
  where each formula is
associated to a weight $w_{kk}>0$ such that $\underset{k\in
  K}{\sum}w_{kk}=1$. Consider $m$ disjoint sets $X_{i}\subset\left\{ 0,1\right\} ^{n}$ and $Z=\left\{ x\oplus y:\left(x,y\right)\in X_{i}\times X_{j}\right\} $, such that $z\in Z$
  implies that $S\left(z\right)_{w}=\frac{1}{2}$, where $S\left(z\right)_{w}$ is the sum of weights of each formula in $\mathbb{X}$ that is satisfied by $z$. The XOR-Weighted-Problem consists in separating sets $X_{i}$ in different outputs.
  
\end{definition}

\begin{corollary}
\label{sinte}

Quantum exact algorithms can solve the XOR-Weighted-Problem within $t+1$ queries.

\end{corollary}
\begin{proof}
This is a reformulation of Theorem~\ref{sma}. Notice that $x_{0}=0$ in $\mathbb{X}$ is a consequence of Eq.~(\ref{oper}).
\end{proof}

Corollary~\ref{sinte} characterizes the power of orthogonal exact algorithms and it represents a model whose complexity upper-bounds the quantum query model. Fig.~\ref{im4} gives a visualization of the XOR-Weighted-Problem. 

\section{A lower bound for exact quantum algorithms}
\label{S6}

In this section, using the BSF approach, we develop a lower bound result for exact quantum  query complexity, considering functions of Boolean domain but arbitrary output.

We apply a basis for the Boolean cube \cite{DEWOLF1}, which is a family of functions
$$\mathcal{F}_{k}^{n}:\left\{ 0,1\right\} ^{n}\rightarrow\left\{ 1,-1\right\},$$ such that each $\mathcal{F}_{k}^{n}(x)=\prod_{i=0}^{n-1} (-1)^{x_{k_{i}}}$ is defined for vectors $k\in\mathbb{Z}_{n}^{n}$.

Consider that, for $k\neq h$, it is possible that $\mathcal{F}_{k}^{n}=\mathcal{F}_{h}^{n}$. Thereby, we need to define an equivalence relation $k\thicksim h$, for $k,h\in\mathbb{Z}_{n}^{n}$ such that $\mathcal{F}_{k}^{n}=\mathcal{F}_{h}^{n}$. We define $\mathbb{Q}_{n}$ as the quotient set of our relation and the set $\left[k\right]\in \mathbb{Q}_{n}$ as
  the equivalence class for element $k$. We also define (a) $\mathbb{F}_{n}$, which elements are defined as functions indexed by $\mathbb{Q}_{n}$ such that $\mathcal{F}_{\left[h\right]}^{n}=\mathcal{F}_{k}^{n}$ iff $k\in\left[h\right]$, and (b) $\mathbb{F}_{n}\left(m\right)\subset F_{n}$, where $\mathcal{F}_{\left[k\right]}\in\mathbb{F}_{n}\left(m\right)$
  iff $\left[k\right]$
  contains an element $h$
  with no more than $2m$
  non-zero terms.  Finally, we define $\mathcal{F}_{\left[k,h\right]}^{n}:\left\{ 0,1\right\} ^{n}\rightarrow\left\{ 0,1\right\} $,
with output 1 iff $\mathcal{F}_{\left[k\right]}^{n}\left(x\right)=\mathcal{F}_{\left[h\right]}^{n}\left(x\right)=1$.
Notice that $$\mathcal{F}_{\left[k,h\right]}^{n}\left(x\right)=\frac{\mathcal{F}_{\left[k\right]}^{n}\left(x\right)+\mathcal{F}_{\left[h\right]}^{n}\left(x\right)+\mathcal{F}_{\left[k\circ h\right]}^{n}\left(x\right)+1}{4},$$
where $\left[k\circ h\right]\in\mathbb{Q}_{n}$ is an equivalence class
such that $\mathcal{F}_{\left[k\circ h\right]}^{n}\left(x\right)=1$
iff $\mathcal{F}_{\left[k\right]}^{n}\left(x\right)=\mathcal{F}_{\left[h\right]}^{n}\left(x\right)$.

For the following result we introduce additional notation. Let $w:\left\{ 0,1\right\} ^{n}\rightarrow\mathbb{R}$ be a function,
and define $w*\mathcal{F}_{\left[h\right]}^{n}=\underset{x\in\left\{ 0,1\right\} ^{n}}{\sum}w\left(x\right)\mathcal{F}_{\left[h\right]}^{n}\left(x\right)$. We denote $\overline{a}$ as a vector such that all its terms are
$a$. Finally, let $\rho\left(i\right)=0$ if $i$ is even, and $\rho\left(i\right)=1$
otherwise.

\begin{theorem}
Consider $m$ disjoint sets $X_{i}\subset\left\{ 0,1\right\} ^{n}$,
such that for each $x\in X_{i}$ there is a set $Z\left(x\right)=\left\{ x\oplus y:y\in X_{j}\, and\, j\neq i\right\} $.
We also define a family of functions $g_{y}^{k}\left(x\right)$ such
that (a) $g_{y}^{k}\left(\overline{0}\right)=1$, (b) $g_{y}^{k}\left(x\right)=\frac{1}{2}$
for $x\in Z\left(y\right)$, (c) $g_{y}^{k}\left(x\right)=1$ for $x$,
where 
\begin{equation}
\label{count}
\underset{i=0}{\overset{2k}{\sum}}\underset{j=0}{\sum}\left(\begin{array}{c}
\left|x\right|\\
i-2j-\rho\left(i\right)
\end{array}\right)\left(\begin{array}{c}
n-\left|x\right|\\
2j+\rho\left(i\right)
\end{array}\right)>\frac{\underset{i=0}{\overset{2k}{\sum}}\left(\begin{array}{c}
n\\
i
\end{array}\right)}{2}
\end{equation}
 and (d) $g_{y}^{k}\left(x\right)=0$ otherwise. If $\underset{\mathcal{F}_{\left[h\right]}^{n}\in\mathbb{F}_{n}\left(k\right)}{\sum}g_{y}^{k}*\mathcal{F}_{\left[h\right]}^{n}\geq1$
for all $y\in\underset{i}{\bigcup}X_{i}$, then an exact quantum algorithm
that gives different outputs for each $X_{i}$, applies at least $k$
queries.
\end{theorem}
\begin{proof}
Suppose that a quantum algorithm allows us to separate $x\in X_{i}$
from $\underset{j\neq i}{\bigcup}X_{j}$,  by applying $k$ queries and
without error. Using Gram matrix representation from Theorem~\ref{thema} at
row $x$, we have
\begin{equation}
\frac{1}{2}\left(G_{\left[x,x\oplus y\right]}+1\right)=
\underset{\left[h\right]\in\mathbb{Q}_{n}}{\sum}\alpha_{\left[h\right]}\mathcal{F}_{\left[h\right]}^{n}\left(y\right)+\underset{\left[h_{i}\right]\neq\left[h_{j}\right]}{\sum}\alpha_{\left[h_{i},h_{j}\right]}\mathcal{F}_{\left[h_{i},h_{j}\right]}^{n}\left(y\right)
\end{equation}
Defining $\overline{T}_{h_{1},h_{2}}=\overline{P}_{h_{1},h_{2}}-2\overline{R}_{h_{1},h_{2}}+\overline{Q}_{h_{1},h_{2}}$
from the matrices in Eq.~(\ref{eqsum2}), notice that first sum in the expression
comes from $\overline{T}_{h_{1},h_{2}}$ when $h_{1}=h_{2}$ and second
sum comes from $\overline{T}_{h_{1},h_{2}}$ when $h_{1}\neq h_{2}$.
Thus, we have$\underset{\left[h\right]\in\mathbb{Q}_{n}}{\sum}\alpha_{\left[h\right]}=1$
and $\underset{\left[h_{i}\right]\neq\left[h_{j}\right]}{\sum}\alpha_{\left[h_{i},h_{j}\right]}=0$,
that implies
\begin{equation}
\underset{\mathcal{F}_{\left[h\right]}^{n}\in\mathbb{F}_{n}\left(k\right)}{\sum}\frac{1}{2}\left(G_{\left[x,x\oplus y\right]}+1\right)*\mathcal{F}_{\left[h\right]}^{n}=\sqrt{2^{n}}.
\end{equation}
Thereby, we can state a necessary condition for an orthogonality between
the final states of $x\in X_{i}$ and $\underset{j\neq i}{\bigcup}X_{j}$,
where the algorithm applies $k$ queries. That is the existence of
some function $g:\left\{ 0,1\right\} ^{n}\rightarrow\left[0,1\right]$,
such that $g\left(x\right)=\frac{1}{2}$ for $x\in\underset{j\neq i}{\bigcup}X_{j}$,
$g\left(\overline{0}\right)=1$ and $\underset{\mathcal{F}_{\left[h\right]}^{n}\in\mathbb{F}_{n}\left(k\right)}{\sum}g*\mathcal{F}_{\left[h\right]}^{n}\geq\sqrt{2^{n}}$.
The function $g_{y}^{k}\left(x\right)$ fulfills such properties maximizing
$\underset{\mathcal{F}_{\left[h\right]}^{n}\in\mathbb{F}_{n}\left(k\right)}{\sum}g*\mathcal{F}_{\left[h\right]}^{n}$. That is, if $x\notin\underset{j\neq i}{\bigcup}X_{j}-\left\{ \overline{0}\right\}$
then $g_{y}^{k}\left(x\right)=1$ for inputs $x$ such that there
are more functions in $\mathbb{F}_{n}\left(k\right)$ with value~$1$
than~$-1$ and $g_{y}^{k}\left(x\right)=0$ otherwise. Notice that
$\underset{i=0}{\overset{2k}{\sum}}\left(\begin{array}{c}
n\\
i
\end{array}\right)$ is the cardinality of $\mathbb{F}_{n}\left(k\right)$ and
$$\underset{i=0}{\overset{2k}{\sum}}\underset{j=0}{\sum}\left(\begin{array}{c}
\left|x\right|\\
i-2j-\rho\left(i\right)
\end{array}\right)\left(\begin{array}{c}
n-\left|x\right|\\
2j+\rho\left(i\right)
\end{array}\right)$$
is the cardinality of functions in $\mathbb{F}_{n}\left(k\right)$
with value~$1$ in $x$.
\end{proof}

This theorem offers an alternative lower-bound to traditional tools like Polynomial and Adversary methods. 
We present a simple example of its application, that is the total function $f$ that separates $X_{1}=\left\{ \overline{0},\overline{1}\right\} $
from $X_{2}=\left\{ 0,1\right\} ^{n}-X_{1}$. For any $y\in X_{1}$,
we have $g_{y}^{k}\left(x\right)=\widehat{g}_{y}\left(x\right)+\widetilde{g}_{y}\left(x\right)$, where (a) $\widehat{g}_{y}\left(\overline{0}\right)=1$ and $\widehat{g}_{y}\left(x\right)=\frac{1}{2}$
for $x\neq\overline{0}$, and (b) $\widetilde{g}_{y}\left(\overline{1}\right)=\frac{1}{2}$
and $\widetilde{g}_{y}\left(x\right)=0$ for $x\neq\overline{1}$.
That is because $k\leq\left\lfloor \frac{n}{2}\right\rfloor $ implies
that Eq.~(\ref{count}) for $x=\overline{1}$ becomes $$\underset{i=0}{\overset{k}{\sum}}\left(\begin{array}{c}
n\\
2i
\end{array}\right)>\frac{\underset{i=0}{\overset{2k}{\sum}}\left(\begin{array}{c}
n\\
i
\end{array}\right)}{2}.$$ 
Thus, we have $$\underset{\mathcal{F}_{\left[h\right]}^{n}\in\mathbb{F}_{n}\left(k\right)}{\sum}\widehat{g}_{y}*\mathcal{F}_{\left[h\right]}^{n}=\frac{1}{\sqrt{2^{n}}}\underset{i=0}{\overset{2k}{\sum}}\left(\begin{array}{c}
n\\
i
\end{array}\right)$$ and $$\underset{\mathcal{F}_{\left[h\right]}^{n}\in\mathbb{F}_{n}\left(k\right)}{\sum}\widetilde{g}_{y}*\mathcal{F}_{\left[h\right]}^{n}=\frac{1}{2\sqrt{2^{n}}}\underset{i=0}{\overset{k}{\sum}}\left(\begin{array}{c}
n\\
2i
\end{array}\right).$$ Choosing $k=\left\lceil \frac{4n}{10}\right\rceil $, we have that
$$\underset{i=0}{\overset{\left\lceil \frac{4n}{10}\right\rceil }{\sum}}\left(\begin{array}{c}
n\\
2i
\end{array}\right)>4\left[\underset{i=2\left\lceil \frac{4n}{10}\right\rceil }{\overset{n}{\sum}}\left(\begin{array}{c}
n\\
i
\end{array}\right)\right].$$ 
That is enough for proving $$\underset{\mathcal{F}_{\left[h\right]}^{n}\in\mathbb{F}_{n}\left(\left\lceil \frac{4n}{10}\right\rceil \right)}{\sum}g_{y}^{\left\lceil \frac{4n}{10}\right\rceil }*\mathcal{F}_{\left[h\right]}^{n}\geq\sqrt{2^{n}},$$ which gives a lower bound $Q_{E}\left(f\right)=\Omega\left(n\right)$.

\section{Conclusion}
\label{S7}
In this work, we presented tree theoretical results.
Our main theoretical result was the Block Set Formulation, which is a
reformulation of the Quantum Query Model such that the unitary
operators are replaced by phase inversions over a set of vectors. This
contribution gives an alternative interpretation on how quantum query
algorithms work. 
A second result is a linear system
of equations that allows an alternative analysis and construction  of quantum exact algorithms for  partial functions. These constructions are delimited by a problem defined by weights over formulas, which can be considered a model that upper-bounds the QQM. Finally, we apply the BSF approach for developing a lower-bound for exact quantum algorithms. These results give a validation of our formulation.

This approach leaves open problems and research possibilities:
\begin{itemize}
\item It is possible to obtain 
  algorithms with some error 
  by using the introduced tools, for example by approximate solutions
  to system $E\left(t,n,X,Y\right)$, but this does not guarantee a
  bounded error.  This approach would be extended by finding a
  sufficient and necessary condition for obtaining a bounded error
  algorithm.
\item The condition $\left(k_{1}\neq k_{2}\Rightarrow
    w_{k_{1}k_{2}}=0\right)$ used in Corollary~\ref{ficor} could be
  weakened for some unknowns obtaining more powerful yet complicated
  models than the orthogonal BSF.  Which strategies can be developed
  for constructing exact quantum algorithms under the general BSF
  (Theorem~\ref{linsys}) 
  or a weaker condition?
\end{itemize}

\subsection*{Acknowledgements}

This work received financial support from CAPES and CNPq. The authors thank
the group of Quantum Computing at LNCC/MCTI and the Laboratory of
Algorithms and Combinatorics at PESC/COPPE/UFRJ for helpful discussions.

\bibliographystyle{plain}
\bibliography{bsfrefs}

\begin{thebibliography}{10}

\bibitem{Aaronson}
Scott Aaronson and Andris Ambainis.
\newblock Forrelation: A problem that optimally separates quantum from
  classical computing.
\newblock In {\em Proceedings of the Forty-Seventh Annual ACM on Symposium on
  Theory of Computing}, pages 307--316. ACM, 2015.

\bibitem{AMB3}
Andris Ambainis.
\newblock Quantum lower bounds by quantum arguments.
\newblock In {\em Proceedings of the thirty-second annual ACM symposium on
  Theory of computing}, pages 636--643. ACM, 2000.

\bibitem{AMB1}
Andris Ambainis.
\newblock Superlinear advantage for exact quantum algorithms.
\newblock {\em SIAM Journal on Computing}, 45(2):617--631, 2016.

\bibitem{Ambainis4}
Andris Ambainis, Kaspars Balodis, Aleksandrs Belovs, Troy Lee, Miklos Santha,
  and Juris Smotrovs.
\newblock Separations in query complexity based on pointer functions.
\newblock arXiv preprint arXiv:1506.04719, 2015.

\bibitem{Ambainis6}
Andris Ambainis, Jozef Gruska, and Shenggen Zheng.
\newblock Exact quantum algorithms have advantage for almost all boolean
  functions.
\newblock arXiv preprint arXiv:1404.1684, 2014.

\bibitem{Ambainis5}
Andris Ambainis, J{\=a}nis Iraids, and Daniel Nagaj.
\newblock Exact quantum query complexity of $\rm{EXACT}_{k,l}^n$.
\newblock arXiv preprint arXiv:1608.02374, 2016.

\bibitem{AMB2}
Andris Ambainis, J{\=a}nis Iraids, and Juris Smotrovs.
\newblock Exact quantum query complexity of exact and threshold.
\newblock arXiv preprint arXiv:1302.1235, 2013.

\bibitem{BARNUM}
Howard Barnum, Michael Saks, and Mario Szegedy.
\newblock Quantum decision trees and semidefinite.
\newblock Technical Report LA-UR-01-6417, Los Alamos National Laboratory, May
  2002.

\bibitem{Beals}
Robert Beals, Harry Buhrman, Richard Cleve, Michele Mosca, and Ronald De~Wolf.
\newblock Quantum lower bounds by polynomials.
\newblock {\em Journal of the ACM (JACM)}, 48(4):778--797, 2001.

\bibitem{Dewolf}
H.~Buhrman and R.~de~Wolf.
\newblock Complexity measures and decision tree complexity: A survey.
\newblock {\em Theoretical Computer Science}, 288(1):21--43, October 1999.

\bibitem{DEWOLF1}
R.~De~Wolf.
\newblock A brief introduction to fourier analysis on the boolean cube.
\newblock Theory of Computing, Graduate Surveys, 2008.

\bibitem{Deutsch}
David Deutsch and Richard Jozsa.
\newblock Rapid solution of problems by quantum computation.
\newblock In {\em Proceedings of the Royal Society of London A: Mathematical,
  Physical and Engineering Sciences}, volume 439, pages 553--558. The Royal
  Society, 1992.

\bibitem{Gruska}
Jozef Gruska, Daowen Qiu, and Shenggen Zheng.
\newblock Generalizations of the distributed deutsch--jozsa promise problem.
\newblock {\em Mathematical Structures in Computer Science}, pages 1--21, 2015.

\bibitem{Hoyer2}
Peter Hoyer, Troy Lee, and Robert Spalek.
\newblock Negative weights make adversaries stronger.
\newblock In {\em Proceedings of the thirty-ninth annual ACM symposium on
  Theory of computing}, pages 526--535. ACM, 2007.

\bibitem{KAYE}
Phillip Kaye, Raymond Laflamme, and Michele Mosca.
\newblock {\em An introduction to quantum computing}.
\newblock Oxford University Press, 2007.

\bibitem{Midrijanis}
Gatis Midrijanis.
\newblock Exact quantum query complexity for total boolean functions.
\newblock arXiv preprint quant-ph/0403168, 2004.

\bibitem{MONTANARO}
Ashley Montanaro, Richard Jozsa, and Graeme Mitchison.
\newblock On exact quantum query complexity.
\newblock {\em Algorithmica}, 71(4):775--796, 2015.

\end{thebibliography}
\section*{Appendix A}
In this appendix, we give a simple one-dimensional example of Block Set associated to a QQM algorithm, namely Deutsch's algorithm. For simplicity our $H_{W}$ is an empty set, thus the algorithm has an initial state $\Psi_{0}=\left|0\right\rangle$. The QQM representation of Deutsch's algorithm takes the unitary operators
\[
U_{0}=\left[\begin{array}{ccc}
0 & 0 & 1\\
\frac{1}{\sqrt{2}} & \frac{1}{\sqrt{2}} & 0\\
\frac{1}{\sqrt{2}} & -\frac{1}{\sqrt{2}} & 0
\end{array}\right],
\] and $U_{1}=I$. The CSOP used by the measurement step is not important for our purposes. Considering the CSOP $\left\{ P_{k}\right\}$ as defined in Section~\ref{S2}, we have $P_{i}=\left|i\right\rangle \left\langle i\right|$. Using Definition~\ref{debs}, we obtain each element of the Block Set, namely
\[\left|\Psi\left(0\right)\right\rangle =\widetilde{P}_{0}^{0}\left|0\right\rangle =U_{0}^{\dagger}P_{0}U_{0}\left|0\right\rangle =\left[\begin{array}{c}
0\\
0\\
0
\end{array}\right],\]
\[\left|\Psi\left(1\right)\right\rangle =\widetilde{P}_{1}^{0}\left|0\right\rangle =U_{0}^{\dagger}P_{1}U_{0}\left|0\right\rangle =\left[\begin{array}{c}
\frac{1}{2}\\
\frac{1}{2}\\
0
\end{array}\right],\] 
and
\[\left|\Psi\left(2\right)\right\rangle =\widetilde{P}_{2}^{0}\left|0\right\rangle =U_{0}^{\dagger}P_{2}U_{0}\left|0\right\rangle =\left[\begin{array}{c}
\frac{1}{2}\\
-\frac{1}{2}\\
0
\end{array}\right].\]

Take $\left\{ \left|\Psi_{x}^{f}\right\rangle \right\}$ 
  and $\left\{ \left|\widetilde{\Psi}_{x}^{f}\right\rangle \right\}$
  as the final states of Deutsch's algorithm and the Block Set, respectively. We consider that $x=x_{n}\ldots x_{2}x_{1}$. Using Definition~\ref{decal} and Theorem~\ref{desc1}, 
  we have
\[\left|\widetilde{\Psi}_{00}^{f}\right\rangle =\left|\Psi\left(0\right)\right\rangle +\left|\Psi\left(1\right)\right\rangle +\left|\Psi\left(2\right)\right\rangle =\left[\begin{array}{c}
1\\
0\\
0
\end{array}\right],\]
\[\left|\widetilde{\Psi}_{01}^{f}\right\rangle =\left|\Psi\left(0\right)\right\rangle -\left|\Psi\left(1\right)\right\rangle +\left|\Psi\left(2\right)\right\rangle =\left[\begin{array}{c}
0\\
-1\\
0
\end{array}\right],\]
\[\left|\widetilde{\Psi}_{10}^{f}\right\rangle =\left|\Psi\left(0\right)\right\rangle +\left|\Psi\left(1\right)\right\rangle -\left|\Psi\left(2\right)\right\rangle =\left[\begin{array}{c}
0\\
1\\
0
\end{array}\right]\]
and 
\[\left|\widetilde{\Psi}_{11}^{f}\right\rangle =\left|\Psi\left(0\right)\right\rangle -\left|\Psi\left(1\right)\right\rangle -\left|\Psi\left(2\right)\right\rangle =\left[\begin{array}{c}
-1\\
0\\
0
\end{array}\right].\]

Since in this case we have the identities \[\left|\Psi_{00}^{f}\right\rangle =\left[\begin{array}{c}
0\\
\frac{1}{\sqrt{2}}\\
\frac{1}{\sqrt{2}}
\end{array}\right]=-\left|\Psi_{11}^{f}\right\rangle \]
  and 
  \[\left|\Psi_{01}^{f}\right\rangle =\left[\begin{array}{c}
0\\
-\frac{1}{\sqrt{2}}\\
\frac{1}{\sqrt{2}}
\end{array}\right]=-\left|\Psi_{10}^{f}\right\rangle,\] if we calculate the Gram matrices of $\left\{ \left|\Psi_{x}^{f}\right\rangle \right\}$
  and $\left\{ \left|\widetilde{\Psi}_{x}^{f}\right\rangle \right\},$
  then we obtain the same matrix 
  \[
G=\left[\begin{array}{cccc}
1 & 0 & 0 & -1\\
0 & 1 & -1 & 0\\
0 & -1 & 1 & 0\\
-1 & 0 & 0 & 1
\end{array}\right].
\] 
In both algorithms, the final states for inputs $X=\left\{ 00,11\right\} $ are orthogonal to the final states for inputs $Y=\left\{ 01,10\right\} $. Thereby there exist CSOPs that discriminate $X$ from $Y$ within error $0$, for both algorithms. This example show that both QQM and BSF algorithms are equivalent to Deutsch's algorithm by choosing the appropriate measurement steps.
\section*{Appendix B}
Here, we extend our previous example of Block Set obtained from Deutsch's algorithm. This extension shows concepts introduced by Section~\ref{secpbs}. In our example, all one-dimensional Block Sets are orthogonal, thus this algorithm is represented by Corollary \ref{bortho2}. In other words, for this case, $k\neq h\Rightarrow$  $\overline{P}_{k,h}+\overline{Q}_{k,h}=0$. Thereby, we are just interested in matrices of the form $\overline{P}_{k,k}+\overline{Q}_{k,k}$. Matrices for each element $k$ are given by
\[
\overline{P}_{0,0}+\overline{Q}_{0,0}=\left[\begin{array}{cccc}
1 & 1 & 1 & 1\\
1 & 1 & 1 & 1\\
1 & 1 & 1 & 1\\
1 & 1 & 1 & 1
\end{array}\right],
\]
\[
\overline{P}_{1,1}+\overline{Q}_{1,1}=\left[\begin{array}{cccc}
1 & 0 & 1 & 0\\
0 & 1 & 0 & 1\\
1 & 0 & 1 & 0\\
0 & 1 & 0 & 1
\end{array}\right]
\] 
and
\[
\overline{P}_{2,2}+\overline{Q}_{2,2}=\left[\begin{array}{cccc}
1 & 1 & 0 & 0\\
1 & 1 & 0 & 0\\
0 & 0 & 1 & 1\\
0 & 0 & 1 & 1
\end{array}\right].
\] Thus, calculating a matrix $M$ with the Block Set obtained in Appendix~A, we get
\[
M=\underset{i=0}{\overset{2}{\sum}}\left(\overline{P}_{i,i}+\overline{Q}_{i,i}\right)\left\langle \Psi\left(i\right)\right.\left|\Psi\left(i\right)\right\rangle =\left[\begin{array}{cccc}
1 & \frac{1}{2} & \frac{1}{2} & 0\\
\frac{1}{2} & 1 & 0 & \frac{1}{2}\\
\frac{1}{2} & 0 & 1 & \frac{1}{2}\\
0 & \frac{1}{2} & \frac{1}{2} & 1
\end{array}\right].
\] 
We finally obtain the Gram matrix of the BSF algorithm from Corollary~\ref{bortho2}. Notice that the resulting matrix is the same as the obtained in Appendix A,
\[
G=2\left(M\right)-\left[\begin{array}{cccc}
1 & 1 & 1 & 1\\
1 & 1 & 1 & 1\\
1 & 1 & 1 & 1\\
1 & 1 & 1 & 1
\end{array}\right]=\left[\begin{array}{cccc}
1 & 0 & 0 & -1\\
0 & 1 & -1 & 0\\
0 & -1 & 1 & 0\\
-1 & 0 & 0 & 1
\end{array}\right].
\] 
This showed how each element of the block set works as a parameter for the Gram matrix of final states. Thus, Eq.~(\ref{bortheq}) is satisfied.

\begin{figure}[p]
    \centering
    \includegraphics[width=0.5\textwidth]{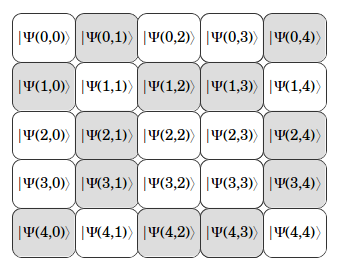}
    \caption{\label{im2} Graphical representation of
    $\widetilde{P}_{x}^{1}\widetilde{P}_{x}^{0}\left|\Psi\right\rangle$ using Eq.~(\ref{desc}) and
    taking as input $x=1001$. Grey boxes represent
    components 
    where the relative phase is inverted with respect to the initial
    state~$\ket{\Psi}$.}
    \label{im2}
\end{figure}
\begin{figure}[p]
    \centering
    \includegraphics[width=0.5\textwidth]{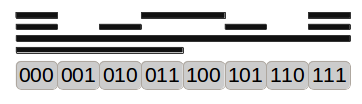}
    \caption{Each black layer represents the
    satisfiability of some formula over an input $x$. In decreasing
    order, the formulas are $x_{1}\oplus x_{2}$, $x_{1}\oplus x_{3}$,
    $x_{0}$ and $x_{3}$. If we give the same weight $\frac{1}{4}$ to
    all these formulas, then any input $x$ having exactly two layers
    over itself is orthogonal to $000$. In our example $001$, $100$
    and $101$.}
    \label{im4}
\end{figure}
\end{document}